\documentclass[11pt]{amsart}

\usepackage{amsthm}
\usepackage{amsfonts}
\usepackage{amssymb}
\usepackage{amsmath}
\usepackage{a4wide}
\usepackage[parfill]{parskip}
\usepackage{color}
\usepackage{url}
\usepackage{graphicx}

\newcommand{\cN}{\mathcal{N}}

\newcommand{\cT}{\mathcal{T}}
\newcommand{\xyz}{\{x,y,z\}}
\newcommand{\xy}{\{x,y\}}
\newcommand{\R}{{\mathbb R}}

\newtheorem{definition}{Definition}
\newtheorem{theorem}{Theorem}

\newtheorem{observation}{Observation}
\newtheorem{corollary}{Corollary}

\title{Trinets encode tree-child and level-2 phylogenetic networks}
\author{Leo van Iersel}
\thanks{Leo van Iersel was supported by a Veni grant of The Netherlands Organisation for Scientific Research (NWO)}
\address{Centrum Wiskunde \& Informatica (CWI)\\ P.O. Box 94079\\
1090 GB Amsterdam, The Netherlands}
\email{l.j.j.v.iersel@gmail.com}
\author{Vincent Moulton}
\address{School of Computing Sciences\\
University of East Anglia\\
Norwich, NR4 7TJ, United Kingdom}
\email{vincent.moulton@cmp.uea.ac.uk}

\begin{document}
\maketitle

\begin{abstract} 
Phylogenetic networks generalize evolutionary trees, and are 
commonly used to represent evolutionary histories of species 
that undergo reticulate evolutionary processes such as
hybridization, recombination and lateral gene transfer.
Recently, there has been great interest in trying to develop
methods to construct rooted phylogenetic networks from
\emph{triplets}, that is rooted trees on three species.
However, although triplets determine 
or \emph{encode} rooted phylogenetic trees, they do not in general encode 
rooted phylogenetic networks, which is
a potential issue for any such method.
Motivated by this fact, Huber and Moulton recently introduced 
\emph{trinets} as a natural extension of rooted 
triplets to networks. In particular, they showed that $\text{level-1}$ 
phylogenetic networks \emph{are} encoded by their trinets, and
also conjectured that all ``recoverable'' rooted 
phylogenetic networks are encoded by their trinets. Here 
we prove that recoverable binary level-2 networks and 
binary tree-child networks are also encoded by 
their trinets. To do this we prove two decomposition 
theorems based on trinets which hold for \emph{all} 
recoverable binary rooted phylogenetic networks. 
Our results provide some additional evidence 
in support of the conjecture that trinets 
encode all recoverable rooted phylogenetic networks, and
could also lead to new approaches to construct
phylogenetic networks from trinets.
\end{abstract}

\section{Introduction}

Phylogenetic trees are routinely used in biology
to represent the evolutionary relationships between
a given set of species. More formally, for a set $X$ 
of species, a \emph{rooted phylogenetic tree} is a 
rooted (graph theoretical) tree that has no 
indegree-1 outdegree-1 vertices, and in which the 
leaves are bijectively labelled by the elements in 
$X$ \cite{SempleSteel2003}; a \emph{(rooted) triplet} 
is a phylogenetic tree with three leaves.
Given a rooted phylogenetic tree $T$ and 
three of its leaves, there is 
unique triplet spanned by those leaves that is contained in~$T$.
A fundamental result in phylogenetics states that 
$T$ is in fact \emph{encoded} by its triplets, that is,
$T$ is the unique phylogenetic tree 
containing the set of triplets 
that arises from taking all combinations of three leaves 
in~$T$~\cite{DressBasic}.
This result is important since it has led to 
various approaches to constructing phylogenetic
trees from set of triplets cf. e.g. \cite{HabibTo2012,lev1athan,RECOMB2008}.

Recently, there has been some interest in 
using networks rather than trees to represent evolutionary
relationships between species that
have undergone reticulate evolution
\cite{HusonRuppScornavacca10,davidbook}. This
is motivated by the fact that processes such as
hybridization, recombination and lateral gene transfer
can lead to evolutionary histories which are not 
best represented by a tree. Formally, a \emph{(rooted phylogenetic) network} 
for a set $X$ of species
is a directed acyclic graph that has a single root, 
has no indegree-1 outdegree-1 vertices, and has 
its leaves bijectively labelled by~$X$ (see Section~\ref{sec:prelim}
for full definitions concerning networks). 
Such a network is called \emph{binary} if all vertices have indegree 
and outdegree at most two and all vertices with indegree 
two have outdegree one. In addition, a binary network is called  
\emph{level}-$k$ \cite{Gambette2009structure,GBP2012,GambetteHuber2012,HabibTo2012,RECOMB2008,simplicityAlgorithmica} if each biconnected component  
has at most~$k$ indegree-2 vertices, and it is
called \emph{tree-child} \cite{CLRV2010b,Cardona2007,ISS2010b,Willson2012b}
if each non-leaf vertex  has at 
least one child which has indegree 1.
Note that a rooted phylogenetic tree 
is a network, but that networks
are more general since they can represent evolutionary
events where species combine rather than speciate.

\begin{figure}
    \centering
    \includegraphics[scale=.7]{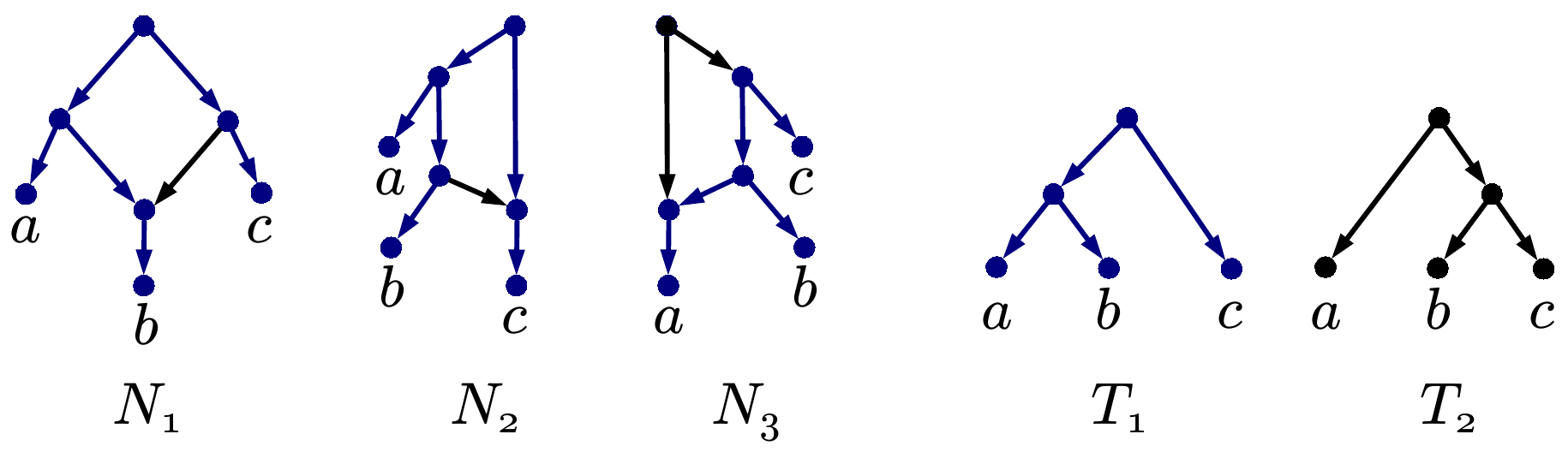}
    \vspace{-.3cm}
   \caption{Three non-isomorphic tree-child, level-1 networks that all have the same set of rooted triplets, that is, $Tr(N_1)=Tr(N_2)=Tr(N_3)=\{T_1,T_2\}$. Blue is used to illustrate how~$T_1$ is contained in~$N_1,N_2$ and~$N_3$. 
}\label{fig:intro}
\end{figure}

As with phylogenetic trees,
efficient algorithms have been developed which, given 
a set of triplets, aim to build a network that contains this set 
(see e.g.~\cite{BGJ10,HabibTo2012,lev1athan,JanssonEtAl2006}). However, these 
algorithms share a common weakness in that, even if
\emph{all} of the triplets within a given network are 
taken as input, there is no guarantee that the
original network will be reconstructed.
This is because, in contrast to trees, the triplets 
in a network do not necessarily
encode the network~\cite{GambetteHuber2012}.
For example, Figure~\ref{fig:intro} presents three different networks 
that all contain the same set of triplets.
Note that a similar observation has been made concerning the
set of trees and set of clusters displayed by a network 
(see e.g.~\cite{HusonRuppScornavacca10,twotrees}).

Motivated by this problem, Huber and Moulton~\cite{huber2011encoding}
recently proposed a possible alternative way to 
encode rooted phylogenetic networks by
introducing a natural extension of rooted triplets to
networks. More specifically, a
\emph{trinet} is a rooted phylogenetic network on three leaves.
As with the triplets in a tree, a network contains or ``exhibits'' a 
trinet on every three leaves (see Section~\ref{sec:prelim}). 
For example, Figure~\ref{fig:networktrinet} presents a 
phylogenetic network and four 
of the trinets that it exhibits. 
The main result in \cite{huber2011encoding} implies that  
$\text{level-1}$  networks encoded by their trinets.
Moreover, it is conjectured that any 
``recoverable'' network (a network 
that satisfies some relatively mild condition which we recall below) 
is also encoded by its trinets. Here, we 
provide some evidence in support of this conjecture 
by showing that recoverable level-2 and tree-child networks 
are also encoded by their trinets. 

\begin{figure}
    \centering
    \includegraphics[scale=.7]{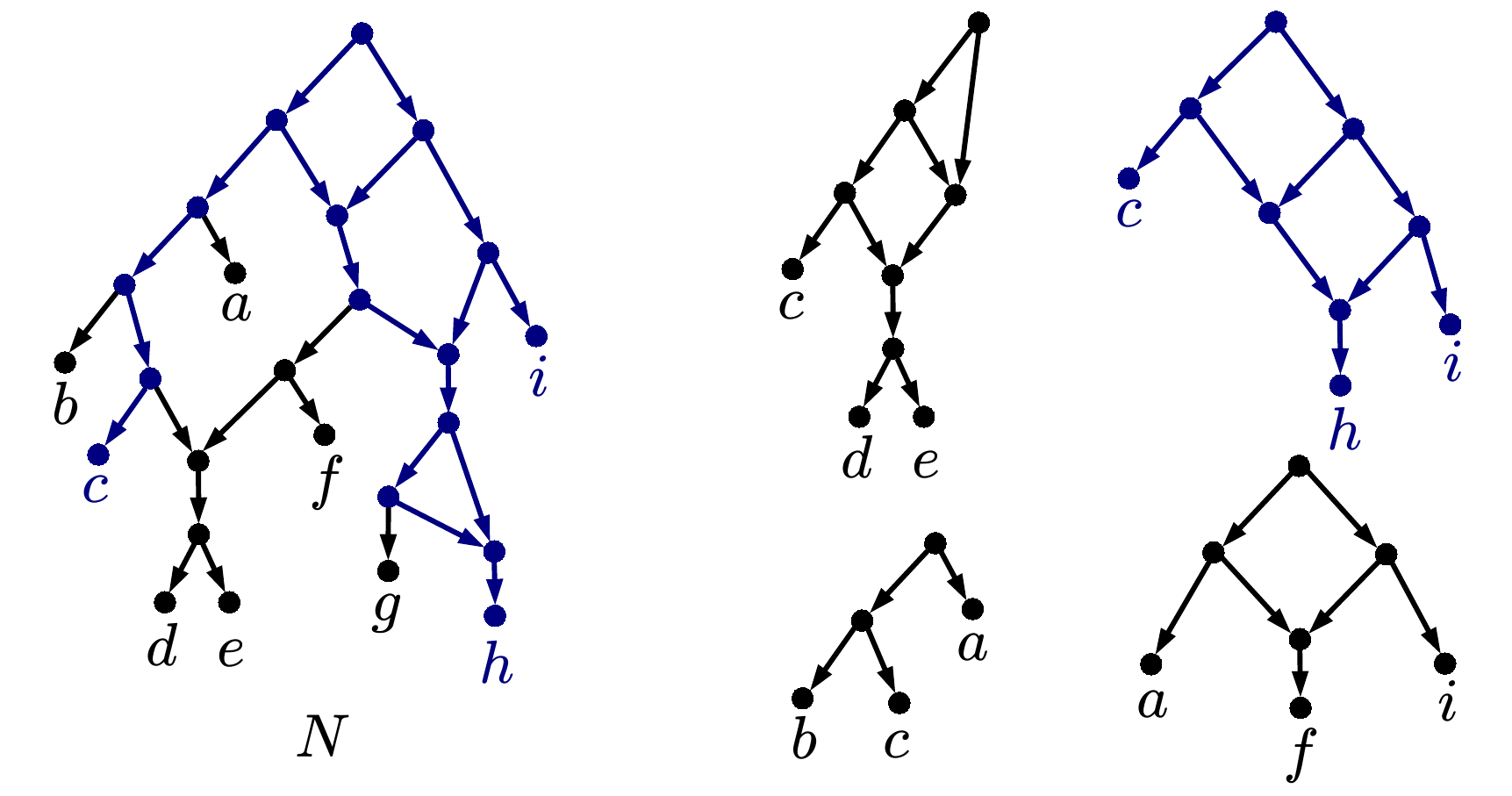}
    \caption{Example of a rooted phylogenetic network~$N$ (left) and four of the trinets exhibited by~$N$ (right). The network~$N$ is binary, recoverable, has level~3 and has the tree-child property. Blue  is used to illustrate how~$N$ exhibits the pictured trinet on~$\{c,h,i\}$.}\label{fig:networktrinet}
\end{figure}

We now give an overview of the rest of 
this paper, in which all networks are assumed to be binary.
After presenting some preliminaries
in Section~\ref{sec:prelim}, we begin by studying the 
relationship between the structure of a 
network and the trinets that it exhibits.
In particular, in Section~\ref{sec:decomp}
we present two decomposition 
theorems for general networks.
Essentially, these two theorems state that
the cut-arcs of a network (that is, arcs 
whose removal disconnect the network) can be directly deduced from 
its set of trinets (Theorem~\ref{thm:caset}), and 
that a network is encoded by its trinets 
if and only if each of its biconnected components
is encoded by its trinets (Theorem~\ref{thm:bcc}). 
In tandem, these theorems essentially restrict the 
problem of deciding whether or not trinets
encode networks to the class of networks
that do not have any cut-arcs apart from pendant arcs 
(so-called ``simple'' networks).

By restricting our attention to simple networks, 
in Section~\ref{sec:lev2} we 
show that a recoverable $\text{level-2}$ network
is always encoded by its trinets (Corollary~\ref{enc-lev2}).
To do this, we use the concept of 
``generators'' for level-$k$ networks, using 
the generators for level-$2$ networks 
presented in \cite{RECOMB2008}.
In Section~\ref{sec:treechild},
we then use alternative techniques
to prove that tree-child networks
are also encoded by their trinets (Theorem~\ref{thm:treechild}).
Note that this class of networks includes the class of 
regular networks \cite{BaroniEtAl2004}. Thus it is 
interesting to  note that
a regular network is encoded by the set 
of trees\footnote{Note that all 
of these trees have the same leaf-set 
as the network.} that it contains~\cite{Willson2010},
but that this is not the case for tree-child networks (e.g. all of the 
networks in Figure~\ref{fig:intro} contain the same set of trees).
We conclude with a discussion of our results, two corollaries, and some 
possible future directions in Section~\ref{sec:postlude}.

Ultimately, it is hoped that the results
presented in this paper
will lead to new methods for constructing phylogenetic
networks. In principle it should be straight-forward to 
infer low-level trinets for 
biological datasets consisting of molecular sequences using 
existing methods to construct phylogenetic networks.
For example, given a multiple sequence alignment, 
the most parsimonious or most likely level-1 or level-2 trinet for
every sub-alignment of three sequences 
could be computed using, e.g., methods
described in \cite{JNST06,JNST09}, 
which becomes computationally tractible since there are a bounded number of 
such trinets (under certain natural restrictions, see Sections~\ref{sec:prelim} and~\ref{sec:postlude}). The structural results in this paper, such as
the decomposition theorems presented in 
Section~\ref{sec:decomp}, could then be used 
to help design algorithms to construct 
networks from the trinets inferred in this way.
Note that this has the potential advantage 
that `breakpoints' need not be computed for the
multiple alignment, a first (and sometimes quite difficult) step 
that is commonly required for constructing phylogenetic 
networks from phylogenetic trees or clusters (cf. 
e.g. \cite[Section 2]{N11}).

\section{Preliminaries}\label{sec:prelim}

Throughout the paper,~$X$ is a finite set. As
mentioned in the introduction, a \emph{rooted phylogenetic network} on~$X$ is a directed acyclic graph with a single indegree-0 vertex (the \emph{root}) and a bijective labelling of its outdegree-0 vertices (\emph{leaves}) by the elements of~$X$. We identify each leaf with its label. A phylogenetic network is \emph{binary} if all vertices have indegree and outdegree at most~2 and all vertices with indegree~2 have outdegree~1. We will often refer to a rooted phylogenetic network simply as a \emph{phylogenetic network} or a \emph{network} for short. See Figures~1,2 and~3 for examples. Let~$u$ and~$v$ be two vertices of a phylogenetic network~$N$. If~$(u,v)$ is an arc of~$N$, then we say that~$u$ is a \emph{parent} of~$v$ and that~$v$ is a \emph{child} of~$u$. Furthermore, we write $u\leq_N v$ and say that~$v$ is \emph{below}~$u$, if there is a directed path from~$u$ to~$v$ in~$N$, or $u=v$. For two leaves~$x$ and~$y$, we say that~$x$ is \emph{below}~$y$ if the parent of~$x$ is below the parent of~$y$. For an arc~$a=(u,v)$ and a vertex~$w$, we say that~$w$ is \emph{below}~$a$ if~$w$ is below~$v$.

Let~$D$ be a directed graph with a single root~$\rho$. The indegree of a vertex~$v$ of~$D$ is denoted $\delta^-(v)$ and~$v$ is said to be a \emph{reticulation vertex} or a \emph{reticulation} if $\delta^-(v)\geq 2$. The \emph{reticulation number} of~$D$ is defined as
\[
r(N) = \sum_{v\neq\rho} (\delta^-(v)-1).
\]
Hence, the reticulation number of a binary network is simply the number of its reticulation vertices.

We say that a vertex~$v$ of~$D$ is a \emph{cut-vertex} if its removal disconnects the underlying undirected graph of~$D$. Similarly, an arc~$a$ of~$D$ is a \emph{cut-arc} if its removal disconnects the underlying undirected graph of~$D$. A directed graph is called \emph{biconnected} if it has no cut-vertices. A \emph{biconnected component} is a maximal biconnected subgraph (i.e. a biconnected subgraph that is not contained in any other biconnected subgraph). Note that, by this definition, each cut-arc is a biconnected component. We call these the \emph{trivial biconnected components}. 
Thus, rephrasing the definitions given in the introduction, a 
phylogenetic network is \emph{level-$k$} if each  biconnected component 
has reticulation number at most~$k$, and 
it is  \emph{tree-child} if every non-leaf vertex of
the network has at least one child that is not a reticulation.

Given a nontrivial biconnected component~$B$, we say that~$B$ is \emph{redundant} if it has only one outgoing arc and we say that~$B$ is \emph{strongly redundant} if it has only one outgoing arc~$(u,v)$ and all leaves of the network are below~$v$. We say that a phylogenetic network~$N$ is \emph{recoverable} if it has no strongly redundant biconnected components (see e.g. Figure~\ref{fig:notrecov}). We remark that all level-1 networks are recoverable~\cite{huber2011encoding}. Moreover, neither level-1 nor tree-child networks can have any redundant biconnected components. On the other hand, there are level-2 networks that \emph{do} have redundant (and strongly redundant) biconnected components (see Figure~\ref{fig:notrecov}).

\begin{figure}
    \centering
    \includegraphics[scale=.7]{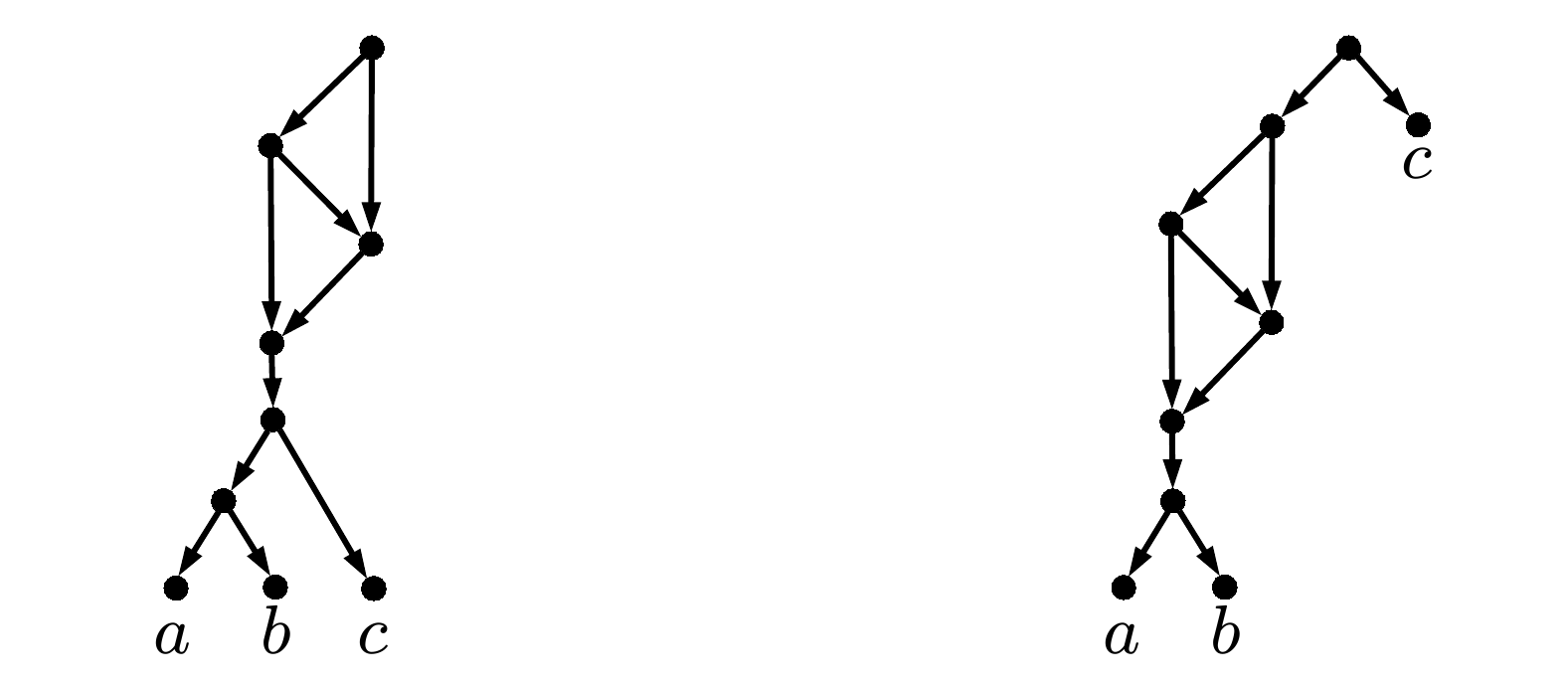}
    \vspace{-.3cm}
    \caption{The phylogenetic network on the left is not recoverable because it has a strongly redundant biconnected component. The phylogenetic network on the right is recoverable, because its only nontrivial biconnected component, although redundant, is not strongly redundant.\label{fig:notrecov}}
\end{figure}

A \emph{trinet} is a phylogenetic network with three leaves. Ignoring leaf-labels, there are~14 distinct level-1 trinets,~8 of which are binary~\cite{huber2011encoding}. Note that there is an infinite number of level-2 trinets, and even of recoverable level-2 trinets. On the other hand, it is not too difficult to see that the number of level-2 trinets without redundant biconnected components is finite (in fact, the number of level-$k$ trinets without redundant biconnected components is bounded by a function of~$k$). We shall return to this point in Section~\ref{sec:postlude}.

Given a network~$N$ on~$X$ and~$X'\subseteq X$, a \emph{lowest stable ancestor} $LSA(X')$ is defined as a vertex~$w\notin X'$ of~$N$ for which all paths from the root to any~$x\in X'$ pass through~$w$, and such that no vertex below~$w$ has this property. A \emph{lowest common ancestor} of~$X'$ in~$N$ is a vertex~$w$ such that~$w\leq_N x$ for all~$x\in X'$ and no vertex below~$w$ has this property. Note that the lowest stable ancestor is unique but that this not necessarily the case for a lowest common ancestor \cite{FH10}. If a lowest common ancestor of~$X'$ \emph{is} unique, then we denote it by $LCA(X')$. For two vertices~$u,v$, we write $LSA(u,v)$ as shorthand for $LSA(\{u,v\})$ and $LCA(u,v)$ as shorthand for $LCA(\{u,v\})$. The following easily proven fact  
will be useful later on.

\begin{observation}\label{obs:lsa}
If~$N$ is a phylogenetic network on~$X$ and~$X'\subseteq X$ with~$X'\geq 2$, then there exist $x,y\in X'$ such that $LSA(x,y)=LSA(X')$.
\end{observation}

Given a phylogenetic network~$N$ on~$X$ and $\xyz\subseteq X$, the trinet on~$\xyz$ \emph{exhibited} by~$N$ is defined as the trinet obtained from~$N$ by deleting all vertices that are not on any path from $LSA(\xyz)$ to~$x$, $y$ or~$z$ and subsequently suppressing all indegree-1 outdegree-1 vertices and parallel arcs. See Figure~\ref{fig:networktrinet}
for some examples. We note that this definition is equivalent to the definition of ``display'' in~\cite{huber2011encoding} but we call it ``exhibit'' to clearly distinguish it from other usages of ``display'' (in particular, the definition of when a network displays a tree or triplet). We will often (implicitly) use the following observation.

\begin{observation}\label{obs:exhibit}
Given a phylogenetic network~$N$ on~$X$ and $\xyz\subseteq X$, the trinet on~$\xyz$ exhibited by~$N$ can be obtained from~$N$ by removing all leaves except~$x,y$ and~$z$ and repeatedly applying the following operations until none is applicable:
\begin{itemize}
\item deleting all unlabelled outdegree-0 vertices;
\item deleting all indegree-0 outdegree-1 vertices;
\item suppressing all indegree-1 outdegree-1 vertices;
\item suppressing all parallel arcs; and
\item suppressing all strongly redundant biconnected components.
\end{itemize}
\end{observation}

\medskip

The following observation, linking lowest common ancestors
in networks and their exhibited trinets, will
be used in the proof of Theorem~\ref{thm:treechild}.

\begin{observation}\label{obs:lcatrinet}
Suppose that~$u$ is the unique lowest common ancestor of two 
leaves~$x$ and~$y$ in a network~$N$ and that~$P$ is a trinet 
exhibited by~$N$ that 
contains~$x$ and~$y$. Then,~$P$ contains~$u$ (where we consider $P$ as
being obtained from $N$ as described in Observation~\ref{obs:exhibit}) 
and~$u$ is the unique lowest common ancestor of~$x$ and~$y$ in~$P$.
\end{observation}

We now make a definition that will be crucial for the decomposition theorems in Section~\ref{sec:decomp}
(note that a somewhat related definition appeared in \cite{HabibTo2012}).

\begin{definition} Let~$N$ be a phylogenetic network on~$X$ and~$A\subseteq X$. Then,~$A$ is a \emph{CA-set (Cut-Arc set)} of~$N$ if there exists a cut-arc~$(u,v)$ of~$N$ such that~$A=\{x\in X\mid v\leq_N x\}$.
\end{definition}

For example, the CA-sets of network~$N$ in Figure~\ref{fig:networktrinet}
are $\{d,e\},\{g,h\}$ and all singletons $\{a\},\{b\},\{c\},\{d\},\{e\},\{f\},\{g\},\{h\},\{i\}$. In the next section, we will make use of the 
following easily proven fact
that relates the CA-sets of a network to the CA-sets of its exhibited trinets.

\medskip

\begin{observation}\label{obs:caset}
Let~$N$ be a phylogenetic network on~$X$ and~$P$ the trinet on~$\xyz\subseteq X$ exhibited by~$N$. If~$A\subseteq X$ is a CA-set of~$N$, then~$A\cap\{x,y,z\}$ is a CA-set of~$P$.
\end{observation}

Given two phylogenetic networks~$N$ and~$N'$ on~$X$, we write $N=N'$ if there is a graph isomorphism between~$N$ and~$N'$ that preserves leaf labels, i.e. if there exists a bijective function $f:V(N)\rightarrow V(N')$ such that $f(x)=x$ for each leaf~$x$ of~$N$ and such that for every $u,v\in V(N)$ holds that $(u,v)$ is an arc of~$N$ if and only if $(f(u),f(v))$ is an arc of~$N'$.

We use~$Tn(N)$ to denote the set of all trinets exhibited by a phylogenetic network~$N$. A phylogenetic network~$N$ is \emph{encoded} by its set of trinets~$Tn(N)$ if there is no recoverable phylogenetic network~$N'\neq N$ with $Tn(N)=Tn(N')$.

\section{Decomposition theorems for trinets}\label{sec:decomp}

It is well known that any graph can be decomposed 
into its biconnected components. 
We begin by showing that trinets can be used to recover 
this decomposition of binary phylogenetic networks. Note that similar 
results have been proven for triplets in~\cite{simplicityAlgorithmica} and
for quartets in unrooted phylogenetics 
networks \cite{Gambette2009structure}. 

\begin{theorem}\label{thm:caset}
Let~$N$ be a recoverable binary phylogenetic network on~$X$, and~$A\subset X$. Then,~$A$ is a CA-set of~$N$ if and only if~$|A|=1$ or, for all $z\in X\setminus A$ and $x,y\in A$ with $x\neq y$, $\xy$ is a CA-set of the trinet on~$\xyz$ exhibited by~$N$.
\end{theorem}
\begin{proof}
Let $A\subset X$. If~$|A|\leq 1$ then the theorem clearly holds. Hence, we assume~$|A|\geq 2$.

To prove the ``only if'' direction, assume that~$A$ is a CA-set of~$N$. Let $z\in X\setminus A$ and $x,y\in A$ with $x\neq y$. There exists a unique trinet~$P$ on~$\xyz$ in~$Tn(N)$. Since~$A$ is a CA-set of~$N$, it follows from Observation~\ref{obs:caset} that $\xy$ is a CA-set of~$P$ and we are done.

It remains to prove the ``if'' direction. Assume that for all $z\in X\setminus A$ and $x,y\in A$ with $x\neq y$, $\xy$ is a CA-set of the trinet on~$\xyz$ exhibited by~$N$. Assume that~$A$ is not a CA-set of~$N$.

First assume that the arc entering $LSA(A)$ is a cut-arc~$a$.  By Observation~\ref{obs:lsa}, there exist~$x,y\in A$ such that $LSA(x,y)=LSA(A)$. Moreover, there exist ${z\in X\setminus A}$ below~$a$ because~$A$ is not a CA-set. Consider the trinet~$P$ on~$\xyz$ exhibited by~$N$. By the definition of ``exhibit'', all paths from $LSA(\xyz)$ and hence from $LSA(x,y)$ to~$x$, $y$ and~$z$ are retained in~$P$. It follows that, $z$ is also below $LSA(x,y)$ in $P$. This means that $\xy$ is not a CA-set of~$P$, which is a contradiction. Hence, there is no cut-arc entering $LSA(A)$, which implies that $LSA(A)$ is in a nontrivial biconnected component.

Now, let~$B$ be the nontrivial biconnected component of~$N$ containing $LSA(A)$. Let~$r_B$ be the root of~$B$. Choose $x,y\in A$ that are below different cut-arcs leaving~$B$ such that $LSA(x,y)=LSA(A)$.
First, we observe that there is no leaf $z\in X\setminus A$ below $LSA(x,y)$, because otherwise we could argue as before that the trinet on~$\xyz$ exhibited by~$N$ does not have~$\xy$ as a CA-set.
Pick $z\in X\setminus A$ arbitrarily below a cut-arc leaving~$B$. Note that neither~$x$ nor~$y$ is below this cut-arc because otherwise~$z$ would be below $LSA(x,y)$.

\begin{figure}
    \centering
    \includegraphics[scale=.7]{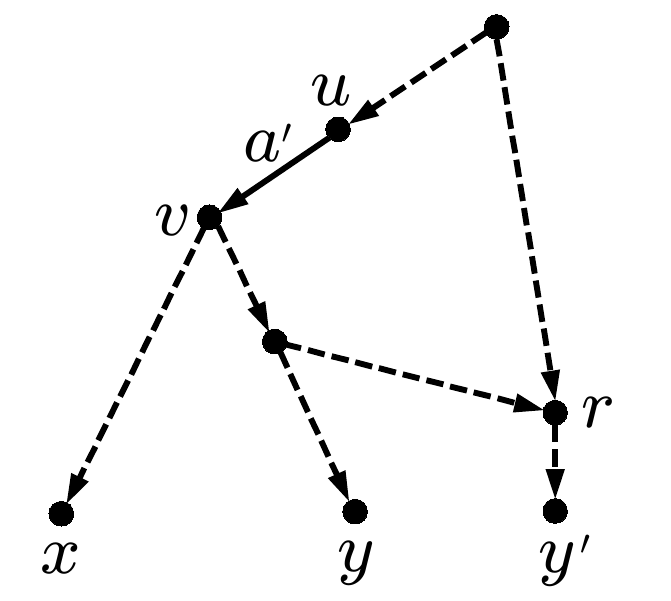}
    \caption{Illustration of network~$N$ in the proof of Theorem~\ref{thm:caset}. Dashed arcs denote directed paths. Arc~$a'$ corresponds to cut-arc~$a$ of trinet~$P$ on~$\xyz$, while~$a'$ is not a cut-arc of~$N$.\label{fig:caset}}
\end{figure}

By assumption, the trinet~$P$ on~$\xyz$ exhibited by~$N$ has $\xy$ as a CA-set. This means that~$P$ has a cut-arc~$a$ such that~$x$ and~$y$ are below~$a$ but~$z$ is not. Consider the arc~$a'=(u,v)$ of~$N$ corresponding to cut-arc~$a$ of~$P$. Observe that~$v$ is the lowest stable ancestor of~$x$ and~$y$ in~$N$, because all paths from $LSA(x,y)$, to~$x$ and~$y$ are retained in~$P$. Also observe that~$a'$ is not a cut-arc in~$N$ because $x,y$ and~$z$ are below three different cut-arcs leaving a biconnected component~$B$. Thus,~$a'$ is some arc of~$B$ and the operations from Observation~\ref{obs:exhibit} that turn~$N$ into~$P$ destroy the biconnectivity of~$B$. Observe that the only one of these operations that does not preserve biconnectivity is the deletion of unlabelled outdegree-0 vertices in case they have indegree greater than~1. We claim that there then exists a reticulation~$r$ in~$N$ with directed paths from~$v$ to~$r$ and from some ancestor of~$u$ to~$r$ not passing through~$a'$ (see Figure~\ref{fig:caset}).

To prove this claim, first note that, since $a'$ is not a cut-arc of~$N$, there is some undirected path~$U$ in~$N$ from~$v$ to~$u$ not passing through~$a'$. Consider the last vertex~$r$ of~$U$ that is below~$v$. Clearly,~$r$ is a reticulation. Let~$p$ be the next vertex on path~$U$. Then~$p$ is not below~$v$. Clearly,~$p$ is below the root and, since~$p$ is not below~$v$, the path from the root to~$p$ does not pass through~$a'$. It follows that~$r$ is a reticulation with directed paths from~$v$ to~$r$ and from some ancestor of~$u$ to~$r$ not passing through~$a'$.

Now, consider any leaf~$y'$ below~$r$. First observe that~$y'\in A$ because no leaves $z\in X\setminus A$ are below $LSA(x,y)$. However, then we obtain a contradiction because $LSA(x,y')$ is closer to the root than $LSA(x,y)=v$.
\end{proof}

Note that we could have used \cite[Lemma~3]{simplicityAlgorithmica}
in the proof of the last result.
However, we presented the above proof since it is 
shorter, self-contained and provides
some insight into how to make arguments using trinets.
We also note that for an arbitrary binary recoverable network~$N$ there is not necessarily a bijection between the cut-arcs of~$N$ and the CA-sets of~$N$ because different cut-arcs might correspond to the same CA-set. However, it is easy to see that the following related observation does hold.

\medskip

\begin{observation}
If~$N$ is a binary phylogenetic network without redundant biconnected components, then there is a bijection between the cut-arcs of~$N$ and the CA-sets of~$N$.
\end{observation}

We now turn to showing that, roughly speaking,  a binary network is encoded by its trinets if and only if each of its biconnected components is encoded by its trinets. To this end, let~$N$ be a phylogenetic network and~$B$ a nontrivial biconnected component with~$b$ outgoing cut-arcs $a_1=(u_1,v_1),\ldots ,a_b=(u_b,v_b)$. Consider the phylogenetic network~$N_B$ obtained from~$N$ by deleting all biconnected components except for~$B,a_1,\ldots ,a_b$ and labelling~$v_1,\ldots ,v_b$ by new labels~$y_1,\ldots ,y_b$ that are not in~$X$. We call~$N_B$ a \emph{restriction} of~$N$ to~$B$. Note that~$N_B$ is unique up to the choice of the new labels $y_1,\ldots ,y_b$.

\medskip

\begin{theorem}\label{thm:bcc}
A recoverable binary phylogenetic network~$N$ on~$X$, with~$|X|\geq 3$, is encoded by its trinets~$Tn(N)$ if and only if, for each nontrivial biconnected component~$B$ of~$N$ with at least four outgoing cut-arcs,~$N_B$ is encoded by $Tn(N_B)$.
\end{theorem}
\begin{proof}
To prove the ``only if'' direction of the theorem, suppose that~$N$ is a recoverable binary phylogenetic network on~$X$ that is encoded by its trinets~$Tn(N)$. Consider any nontrivial biconnected component~$B$ of~$N$ with at least four outgoing cut-arcs. For contradiction, suppose that~$N_B$ is not encoded by~$Tn(N_B)$, i.e. there exists a recoverable network~$N_B'\neq N_B$ such that $Tn(N_B)=Tn(N_B')$. By Theorem~\ref{thm:caset},~$N_B'$ has the same CA-sets as~$N_B$. Moreover, since $Tn(N_B)=Tn(N_B')$ and~$N_B$ has no redundant biconnected components, it follows quite easily that~$N_B'$ has no redundant biconnected components. Combining these observations, we see that~$N_B'$ consists of one nontrivial biconnected component with leaves attached to it by cut-arcs. Let~$B'$ be the nontrivial biconnected component of~$N_B'$. Let~$N'$ be the result of replacing~$B$ by~$B'$ in~$N$. We will show that~$Tn(N)=Tn(N')$, which will contradict the fact that~$N$ is encoded by~$Tn(N)$, since~$N'$ is clearly recoverable.

To show that~$Tn(N)=Tn(N')$, let~$P\in Tn(N)$ and let~$x,y$ and~$z$ be the leaves of~$P$. If~$x,y$ and~$z$ are all below different cut-arcs, or all below the same cut-arc leaving~$B$, then clearly $P\in Tn(N')$ since the only difference between~$N$ and~$N'$ is that~$B$ is replaced by~$B'$, and~$Tn(N_B)=Tn(N_B')$. Now suppose that~$P$ contains leaves~$x,y$ that are below the same cut-arc leaving~$B$ and a leaf~$z$ below a different cut-arc leaving~$B$. Then consider a fourth leaf~$q$ that is below a third cut-arc leaving~$B$. Since $Tn(N_B)=Tn(N_B')$, the trinets on~$\{x,z,q\}$ exhibited by~$N$ and~$N'$ are the same. Hence, the binet (phylogenetic network on two leaves) on~$\{x,z\}$ exhibited (defined in the same way as for trinets) by~$N$ and by~$N'$ is the same. Hence, the trinet on~$\xyz$ exhibited by~$N$ and by~$N'$ is the same, and so $P\in Tn(N')$. The case that~$P$ contains one or more leaves that are not below~$B$ can be handled similarly. It therefore easily follows that $Tn(N)=Tn(N')$, as required.

To prove the ``if'' direction, let~$N$ be a recoverable phylogenetic network on~$X$ such that for each nontrivial biconnected component~$B$ with at least four outgoing cut-arcs the network $N_B$ is encoded by $Tn(N_B)$. Let~$N'$ be a recoverable network on~$X$ with~$Tn(N)=Tn(N')$. We will show that~$N=N'$.

First observe that, for a biconnected component~$B$ with precisely~3 outgoing cut-arcs,~$N_B$ is trivially encoded by $Tn(N_B)$, since in that case~$N_B$ is isomorphic to the single trinet in $Tn(N_B)$.

The rest of the proof is by induction on~$|X|$. If~$|X|=3$, then, since~$N$ and~$N'$ are recoverable, they are both equal to the single trinet in~$Tn(N)$ and we are done. Assume $|X|\geq 4$. Consider the root~$\rho$ of~$N$. We shall assume that~$\rho$ is in some nontrivial biconnected component~$B_\rho$
and that~$a_1=(u_1,v_1),\ldots ,a_b=(u_b,v_b)$ are the cut-arcs leaving~$B_\rho$. The case that~$\rho$ is not in a nontrivial biconnected component can be handled in a similar way, with arcs $a_1,\ldots ,a_b$ being the arcs leaving~$\rho$ (and $b=2$ since~$N$ is binary).

Let $N_1,\ldots ,N_b$ be the networks rooted at $v_1,\ldots ,v_b$. More precisely, for $1\leq i\leq b$, let~$N_i$ be the network obtained from~$N$ by deleting all vertices that are not below~$v_i$. Suppose that~$X_i$ is the leaf-set of~$N_i$. Then, since $b\geq 2$, we have $|X_i|<|X|$. Note that~$N_i$ is not necessarily recoverable.

Now, by Theorem~\ref{thm:caset},~$N'$ has the same CA-sets as~$N$. Thus, $X_i$ is a CA-set of~$N'$ for~$i=1,\ldots ,b$. Since the root~$\rho$ of~$N$ is in some nontrivial biconnected component~$B_\rho$, it follows quite easily that also the root~$\rho'$ of~$N'$ is in some nontrivial biconnected component~$B_\rho'$. Let~$a_1'=(u_1',v_1'),\ldots ,a_b'=(u_b',v_b')$ be the cut-arcs leaving~$B_\rho'$. Let $N_1',\ldots ,N_b'$ be the networks rooted at $v_1',\ldots ,v_b'$. Assume without loss of generality that~$N_i'$ is a network on~$X_i$ for $i=1,\ldots ,b$. To show that~$N=N'$, it remains to show that~$N_{B_\rho}=N_{B_\rho'}$ and that~$N_i=N_i'$ for $i=1,\ldots ,b$.

First, we show that $N_{B_\rho}=N_{B_\rho'}$. If~$b\geq 4$, this is true by assumption (because $Tn(N_{B_\rho})=Tn(N_{B_\rho'})$ and by assumption $N_{B_\rho}$ is encoded by $Tn(N_{B_\rho}$). Moreover, $b\geq 2$ since~$N$ is recoverable. For~$b=3$ the statement is trivial. Hence, the only case left is~$b=2$. Consider two leaves~$x,y$ of~$N$ that are below the same cut-arc leaving~$B_\rho$ and a leaf~$z$ that is below the other cut-arc leaving~$B_\rho$. These leaves exist since~$|X|\geq 3$. Consider the trinet~$P$ in~$Tn(N)$ on~$\xyz$. Let~$B_\rho(P)$ be the biconnected component of~$P$ containing the root of~$P$. Then, $N_{B_\rho(P)}=N_{B_\rho}$. Moreover, since~$N'$ also exhibits~$P$, $N_{B_\rho(P)}=N_{B_\rho'}$. It follows that $N_{B_\rho}=N_{B_\rho'}$.

Now, let~$i\in\{1,\ldots ,b\}$. We will show that $N_i=N_i'$. Since~$|X_i|<|X|$, this follows by induction if (a)~$N_i$ and~$N_i'$ are recoverable and~(b) $|X_i|\geq 3$. To show the general case, consider the networks~$R_i$ and~$R_i'$ obtained from~$N_i$ and~$N_i'$ respectively by suppressing all strongly redundant biconnected components. Then,~$R_i$ and~$R_i'$ are recoverable. Hence, if~$|X_i|\geq 3$, $R_i=R_i'$ by induction. If~$|X_i|=1$, then clearly $R_i=R_i'$ because both consist of a single leaf. The only case left is $|X_i|=2$. Consider any leaf~$z\in X\setminus X_i$ and the trinet~$P$ on $X_i\cup\{z\}$. By Observation~\ref{obs:caset}, $X_i$ is a CA-set of~$P$. Let~$P^*$ be the result of deleting all vertices that are not below~$LSA(X_i)$. Then, $P^*=R_i$. Moreover, since~$P$ is exhibited by~$N'$, we also have $P^*=R_i'$. Hence, in all cases, $R_i=R_i'$. So, to complete the proof that $N_i=N_i'$, it remains to show that~$N_i$ and~$N_i'$ have the same strongly redundant biconnected components, in the same order. To prove this, we distinguish the cases $|X_i|=1$ and $|X_i|\geq 2$.

First, suppose $|X_i|=1$, say~$X_i=\{x\}$. Let~$y,z\in X\setminus X_i$ such that $LSA(u_i)\leq_N z$. Consider the trinet~$P$ on~$\xyz$. Let~$a$ be the cut-arc in~$P$ such that~$x$ is below~$a$,~$y$ and~$z$ are not below~$a$ and there is no cut-arc~$a'$ with this property with~$a$ below~$a'$. Consider the network~$P_x$ obtained from~$P$ by deleting all vertices that are not below~$a$. Then, $P_x=N_i=N_i'$.

Now suppose that $|X_i|\geq 2$. Let~$z\in X\setminus X_i$ such that $LSA(u_i)\leq_N z$ and let~$x,y\in X_i$ such that $LSA(x,y)=LSA(X')$ (such~$x,y$ exist by Observation~\ref{obs:lsa}). Consider the trinet~$P$ on~$\xyz$. Consider the cut-arc~$a$ of~$P$ such that~$x$ and~$y$ are below~$a$,~$z$ is not below~$a$ and such that there is no cut-arc~$a'$ with this property with~$a$ below~$a'$. Let $D$ be the directed graph obtained from~$P$ by deleting all vertices that are not below~$a$ and deleting all vertices that are below $LSA(x,y)$. Then, $D$ is isomorphic to the strongly redundant biconnected components of~$N_i$ and of~$N_i'$. Now, since $R_i=R_i'$ and~$N_i$ and~$N_i'$ have the same strongly redundant biconnected components, in the same order, as required.
\end{proof}

\section{Trinets encode level-2 networks}\label{sec:lev2}

In this section we show that binary recoverable level-2 networks are
encoded by their trinets. To do this, we will 
consider each biconnected component of such a network separately,
and will apply some structural results concerning these components
that are presented in \cite{RECOMB2008}. 
Throughout the section, we restrict to binary networks.

\begin{figure}
    \centering
    \includegraphics[scale=.6]{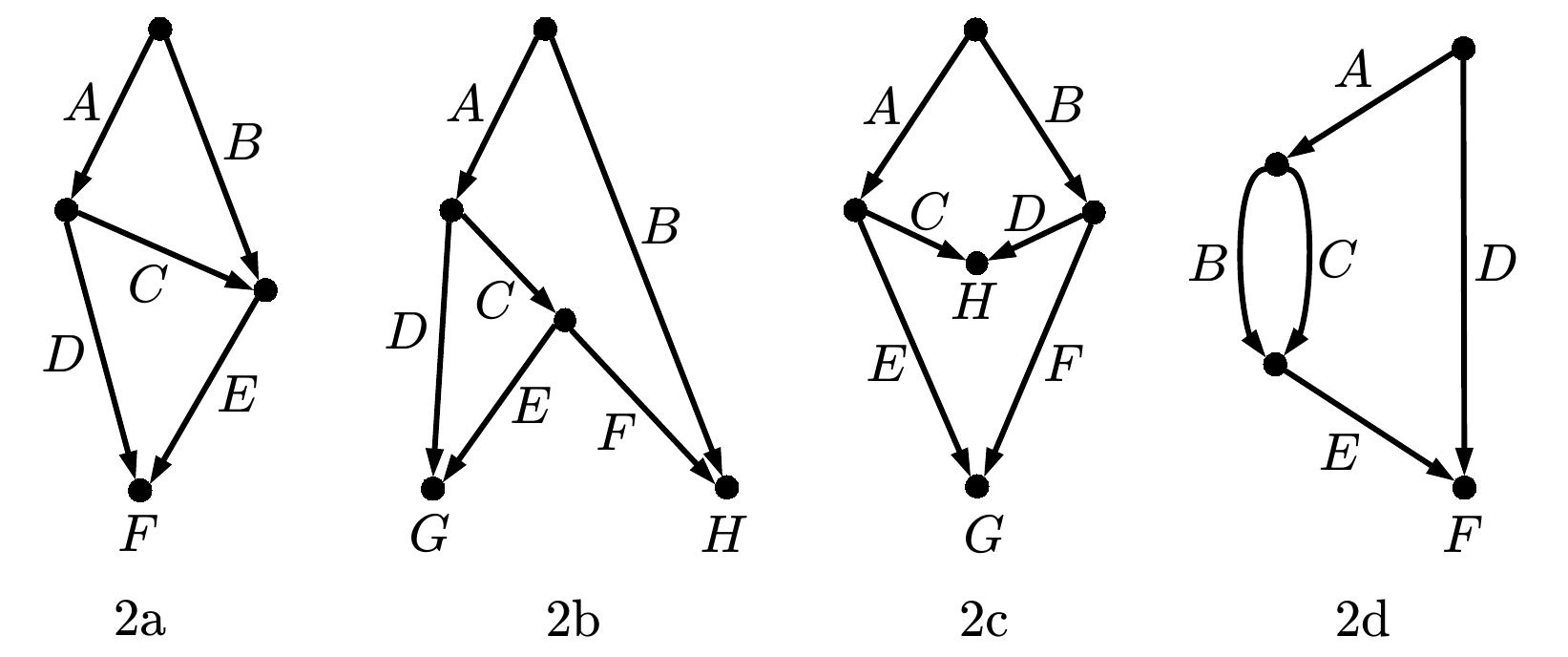}
    \caption{The four level-2 generators. Each side is labelled by a capital letter.\label{fig:generators}}
\end{figure}

We begin by recalling some relevant definitions.
A level-$k$ phylogenetic network is called a \emph{simple level}-$k$ network if it contains one nontrivial biconnected component~$B$ containing exactly~$k$ reticulations and no cut-arcs other than the ones leaving~$B$
(see the left of Figure~\ref{fig:case2b} for an example of a simple level-2 network). A (binary) level-$k$ \emph{generator} is a directed acyclic biconnected multigraph with exactly~$k$ reticulations with indegree~2 and outdegree at most~1, a single vertex with indegree~0 and outdegree~2, and apart from that only vertices with indegree~1 and outdegree~2. The arcs and outdegree-0 vertices of a generator are called its \emph{sides}. For example, all level-2 generators are depicted in Figure~\ref{fig:generators}.

Note that deleting all leaves of a simple level-$k$ network~$N$ gives a level-$k$ generator~$G_N$. We call~$G_N$ the \emph{underlying generator} of~$N$.
Conversely,~$N$ can be reconstructed from~$G_N$ by ``hanging leaves'' from the sides of~$G_N$ as follows (see van Iersel et al.~\cite{RECOMB2008}):
\vspace{-.2cm}
\begin{itemize}
\item for each arc~$a$ of~$G_N$, replace~$a$ by a directed path with $\ell\geq 0$ internal vertices $v_1,\ldots ,v_\ell$ and, for each such internal vertex~$v_i$, add a leaf~$x_i\in X$ and an arc $(v_i,x_i)$; and
\item for each indegree-2 outdegree-0 vertex $v$, add a leaf~$x\in X$ and an arc $(v,x)$.
\end{itemize}
We say that a leaf~$x$ ``is on side'' $s$ if it is hung on side~$s$ in this construction of~$N$ from~$G_N$. More precisely, for a leaf $x\in X$ of a simple level-$k$ network~$N$ with underlying generator~$G_N$ and a side~$s$ of~$G_N$, we say that~$x$ \emph{is on side}~$s$ if~$s$ is an indegree-2 outdegree-0 vertex of~$G_N$ and~$(s,x)$ is an edge of~$N$ or if~$s$ is an edge~$(u,v)$ of~$G_N$ and the parent of~$x$ in~$N$ lies on the directed path from~$u$ to~$v$ in~$N$.

Now, given a level-2 generator~$G$, we call a set of sides of~$G$ a \emph{set of crucial sides} if it contains all vertices with indegree~2 and outdegree~0 together with one arc of each pair of parallel arcs. Consider any simple level-2 network~$N$ on~$X$ with underlying generator~$G$ and a trinet~$P$ on $X'\subseteq X$. We say that~$P$ is a \emph{crucial trinet} of~$N$ if~$X'$ contains at least one leaf on each side in some set of crucial sides of~$G$. For example, Figure~\ref{fig:case2b} depicts a simple level-2 network, one crucial trinet and two non-crucial trinets. The following observation can be verified by inspecting all level-2 generators in Figure~\ref{fig:generators}.

\medskip

\begin{observation}\label{obs:atleastonecrucial}
If~$G$ is a level-2 generator, then it has a set of crucial sides of size at most~2. Hence, every simple level-2 network~$N$ has at least one crucial trinet. Moreover, for every leaf~$x$ of~$N$, there exists a crucial trinet of~$N$ containing~$x$.
\end{observation}

Before proving the main result of this section, we also state one other useful fact.
\medskip

\begin{observation}\label{obs:crucial}
Let~$N$ be a simple level-$k$ network,~$G$ its underlying generator and~${P\in Tn(N)}$. Then,~$P$ is a crucial trinet of~$N$ if and only if~$P$ is a simple level-$k$ network. Moreover, if~$P$ is a crucial trinet of~$N$ then~$G$ is its underlying generator.
\end{observation}

\medskip

\begin{theorem}\label{lem:simplelevel2}
Every binary, simple level-2 network on~$X$, with~$|X|\geq 3$, is encoded by its trinets.
\end{theorem}
\begin{proof}
Let~$N$ be any binary, simple level-2 network on~$X$, with~$|X|\geq 3$. Assume that~$Tn(N')=Tn(N)$ for some recoverable network~$N'$. We will show that~$N'=N$.

We begin by showing that~$N'$ is a binary, simple, level-2 network. First, it is a level-2 network because any level-$k$ network with $k>2$ has a level-$k'$ trinet with $k'>2$, but $Tn(N')=Tn(N)$ contains only level-2 trinets. Second, $N'$ is simple network because its set of CA-sets equals the set of CA-sets of~$N$ by Theorem~\ref{thm:caset} and it has no redundant biconnected components because the trinets in $Tn(N')=Tn(N)$ have no redundant biconnected components. Third, $N'$ is binary. Indeed, assume that~$N'$ has a vertex with outdegree greater than~2 and let~$c_1,c_2,c_3$ be three of its children. Then, consider three (not necessarily different) leaves $x_1,x_2$ and~$x_3$ below $c_1,c_2$ and~$c_3$ respectively. Then, any trinet containing $x_1,x_2$ and~$x_3$ exhibited by~$N'$ is not binary, while all trinets in $Tn(N')=Tn(N)$ are binary.

Now, let~$G$ be the underlying generator of~$N$. First, we show that~$G$ is also the underlying generator of~$N'$. By Observation~\ref{obs:atleastonecrucial},~$N$ has at least one crucial trinet~$P_c$. By Observation~\ref{obs:crucial},~$P_c$ is a simple level-2 network and its underlying generator is~$G$. Since $Tn(N)=Tn(N')$, $P_c$ is also a trinet of~$N'$. Moreover,~$P_c$ is a crucial trinet of~$N'$ by Observation~\ref{obs:crucial} because~$N'$ is a simple level-2 network and~$P_c$ is a simple level-2 network. Hence,~$G$ is the underlying generator of~$N'$, again by Observation~\ref{obs:crucial}.

The remainder of the proof is divided into four cases, based on the four level-2 generators~$2a$, $2b$, $2c$ and~$2d$ (see Figure~\ref{fig:generators} for these generators and the labels of their sides).

Case $G=2a$. First, observe that there are no symmetries, i.e. no relabelling of the sides of~$2a$ gives an isomorphic generator. Let~$x$ be the leaf on side~$F$ in~$N$. Since~$x$ is then the leaf on side~$F$ in every crucial trinet of~$N$, and since these crucial trinets are exhibited by~$N'$, and since there are no symmetries, it follows that~$x$ is also the leaf on side~$F$ in~$N'$. Now consider any side~$s\neq F$ of~$N$ and any leaf~$y$ on that side. Consider any crucial trinet~$P_c$ of~$N$ containing~$y$. Then~$y$ is on side~$s$ in~$P_c$ and, since~$P_c$ is exhibited by~$N'$ and there are no symmetries, $y$ is on side~$s$ in~$N'$. Hence, each leaf is on the same side in~$N'$ as it is in~$N$. It remains to show that the leaves on each side are in the same order in~$N$ and~$N'$. Consider a side~$s$ with at least two leaves and two leaves~$y,z$ on that side such that~$z$ is below~$y$. It follows that $z$ is below~$y$ in the crucial trinet on~$\xyz$ and from that it follows that $z$ is below~$y$ in~$N'$. We conclude that~$N'=N$ since both networks have the same underlying generator, the same leaves on each side, and the same order of the leaves on each side.

\begin{figure}
    \centering
    \includegraphics[scale=.7]{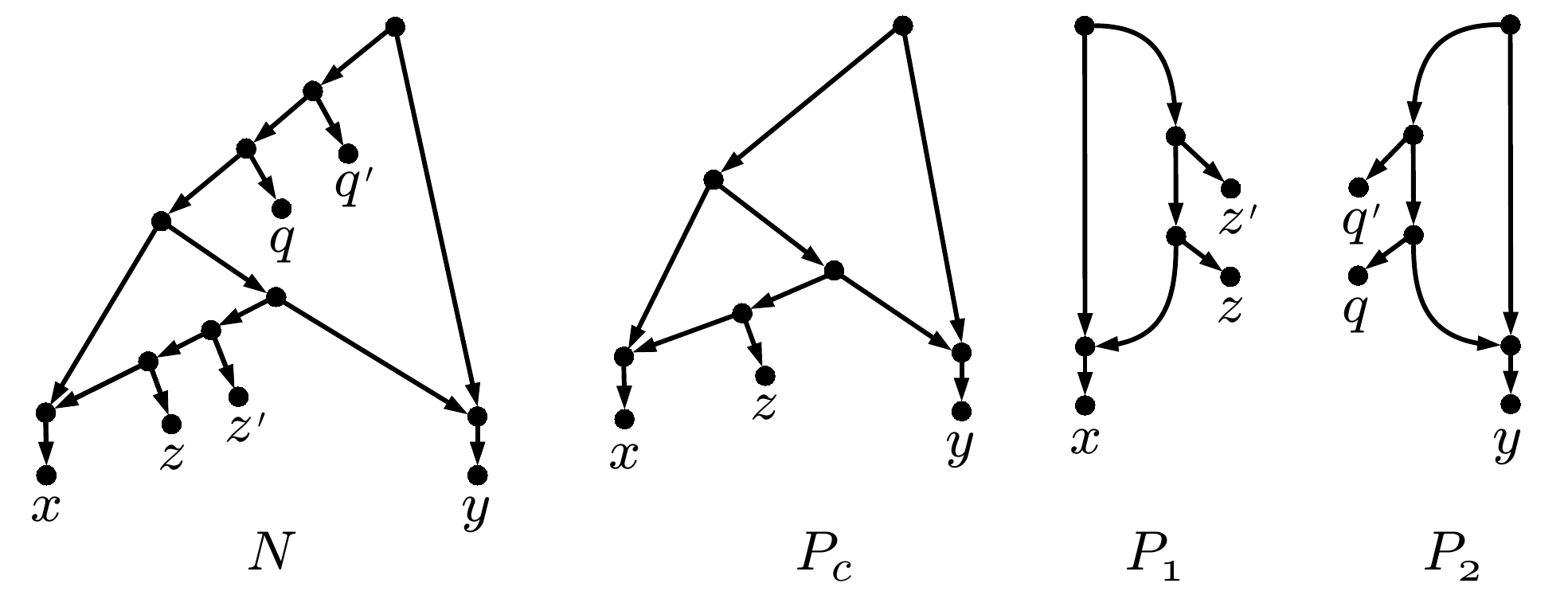}
    \vspace{-.3cm}
    \caption{The underlying generator of simple level-2 network~$N$ is~$2b$ (see Figure~\ref{fig:generators}). Leaf~$x$ is on side~$G$, $y$ is on side~$H$, $z$ and~$z'$ are on side~$E$ and~$q$ and~$q'$ are on side~$A$. Trinet~$P_c$, the trinet on~$\xyz$, is one of the four crucial trinets, which determine the side each leaf is on. Trinet~$P_1$ on~$\{x,z,z'\}$ and trinet~$P_2$ on~$\{y,q,q'\}$, which are non-crucial trinets, determine the order of the leaves on each side. This is Case ``$G=2b$'' in the proof of Lemma~\ref{lem:simplelevel2}.\label{fig:case2b}}
\end{figure}

Case $G=2b$. Again, there are no symmetries. Let~$x$ be the leaf on side~$G$,~$y$ the leaf on side~$H$ and~$z$ a leaf on some other side~$s$ (see Figure~\ref{fig:case2b}). Then, the trinet~$P_c$ on~$\xyz$ is crucial and, since there are no symmetries, it follows that leaves $x,y,z$ are, respectively, on sides~$G,H,s$ in~$P_c$ and hence in~$N'$. Consequently, all leaves are on the same side in~$N'$ as in~$N$. To see that they are in the same order, first consider two leaves~$z,z'$ that are both on side $C,D$ or~$E$ and consider the (non-crucial) trinet~$P_1$ on $\{x,z,z'\}$. Observe that~$P_1$ is a simple level-1 network and that~$z$ and~$z'$ are on the same side of~$P_1$. Moreover, if~$z$ is below~$z'$ in~$N$, then~$z$ is below~$z'$ in~$P_1$, and hence~$z$ is below~$z'$ in~$N'$. Now consider leaves $q,q'$ both on side~$A,B$ or~$F$. Then the trinet~$P_2$ on~$\{y,q,q'\}$ is a simple level-1 network and, as before, if~$q$ is below~$q'$ in~$N$, then~$q$ is below~$q'$ in~$P_2$ and hence in~$N'$. It follows that $N=N'$ as required.

Case $G=2c$. In this case there is some symmetry since sides~$A,C$ and~$E$ can be interchanged with~$B,D$ and~$F$, respectively, to obtain an isomorphic generator. Similarly, sides $C,H,D$ can be interchanged with~$E,G,F$, respectively, again yielding an isomorphic generator. Let~$x$ be on side~$G$, $y$ on side~$H$ and~$z$ on some other side~$s$ in~$N$. Then, the crucial trinet~$P_c$ on~$\xyz$ implies that $x$ and~$y$ are on side~$G$ and~$H$ in~$N'$. Assume without loss of generality that~$x$ is on side~$G$ and~$y$ on side~$H$ in~$N'$. Then, again using trinet~$P_c$, it follows that~$z$ is on side~$A$ or~$B$ in~$N'$ if it is on side~$A$ or~$B$ in~$N$. Similarly,~$z$ is on side~$C$ or~$D$ in~$N'$ if it is on side~$C$ or~$D$ in~$N$ and~$z$ is on side~$E$ or~$F$ in~$N'$ if it is on side~$E$ or~$F$ in~$N$. 

Now, consider two leaves~$z,z'$ that are both on side~$A,B,C$ or~$D$. In view of the trinet on $\{y,z,z'\}$, $z$ and~$z'$ are on the same side of~$N'$ and in the same order. Similarly, for two leaves~$z,z'$ that are both on side~$E$ or~$F$. Also, the trinet on $\{x,z,z'\}$ implies that $z$ and~$z'$ are on the same side of~$N'$ and in the same order. Thus, leaves that are on the same side in~$N$ are on the same side in~$N'$ and in the same order. First assume that there is at least one leaf on side~$A$ in~$N$ and that the leaves that are on side~$A$ in~$N$ are on side~$A$ in~$N'$. Let~$a$ be one such leaf on side~$A$. Then, any leaf~$c$ that is on side~$C$ in~$N$ is on side~$C$ in~$N'$ by the trinet on~$\{a,c,y\}$ (because~$a$ and~$c$ are on the same side of this trinet, which is a simple level-1 network). Similarly, for leaf~$z$ on side $s\in\{B,D\}$ in~$N$ holds that~$z$ is on side~$s$ in~$N'$ by the trinet on $\{a,z,y\}$ and for leaf~$z$ on side $s\in\{E,F\}$ in~$N$ holds that~$z$ is on side~$s$ in~$N'$ by the trinet on $\{a,z,x\}$. It follows that $N=N'$ because all leaves are on the same side, in the same order. Now assume that the leaves that are on side~$A$ in~$N$ are not on side~$A$ in~$N'$. Then these leaves are on side~$B$ in~$N'$. Then we can argue in exactly the same way that the leaves that are on sides $B,C,D,E,F$ in~$N$ are on sides $A,D,C,F,E$ in~$N'$. Hence, again $N=N'$ by relabelling the sides appropriately. Finally, if there is no leaf on side~$A$, then there is a leaf on one of the sides $B,C,D,E,F$ (since $|X|\geq 3$) and we can apply similar arguments based on that leaf.

Case $G=2d$. In this case, the only symmetry is that sides~$B$ and~$C$ can be interchanged with~$C$ and~$B$, respectively. Let~$x$ be the leaf on side~$F$, $y$ a leaf on side~$B$ or~$C$ and~$z$ a leaf on some side~$s\in\{A,B,C,D,E\}$ in~$N$. Note that there exists at least one such leaf~$z$ since $|X|\geq 3$. Then, by the crucial trinet on~$\xyz$, $x$ is on side~$F$ and $y$ is on side~$B$ or~$C$. Without loss of generality,~$y$ is on the same side in~$N'$ as in~$N$. So it follows that~$z$ is on side~$s$ in~$N'$. Hence, without loss of generality (i.e. by relabelling sides~$B$ and~$C$ if necessary), each leaf is on the same side in~$N'$ as in~$N$. Now consider two leaves~$z,z'$ that are on the same side of~$N$. Then the trinet on $\{x,z,z'\}$ implies that the order of~$z$ and~$z'$ is the same in~$N'$ as in~$N$. We can conclude that~$N'=N$, since (after possibly relabelling sides~$B$ and~$C$) both networks have the same leaves on the same sides in the same order.
\end{proof}

\medskip

\begin{corollary}\label{enc-lev2}
Every binary recoverable level-2 network~$N$ on~$X$, with~$|X|\geq 3$, is encoded by its set of trinets~$Tn(N)$.
\end{corollary}
\begin{proof}
Follows from Theorem~\ref{thm:bcc}, Lemma~\ref{lem:simplelevel2} and the fact that level-1 networks are encoded by their trinets~\cite{huber2011encoding}.
\end{proof}

\section{Trinets encode tree-child networks}\label{sec:treechild}

In this section we show that 
tree-child networks are encoded by 
their trinets. We begin by presenting a
definition and some observations.
A directed path in a network is called a \emph{tree path} if it 
does not contain any reticulations apart from possibly 
its first vertex. It is easily seen that from every 
vertex of a tree-child network there is a directed 
tree path that ends at some leaf. 

\medskip

\begin{observation}\label{obs:lcaunique}
Suppose that a network~$N$ has an arc~$(u,v)$ such that~$v$ is a reticulation and such that there is no directed path from~$u$ to the other parent of~$v$. Suppose that there are tree paths from~$u$ to a leaf~$x$ and from~$v$ to a leaf~$y$. Then,~$x$ and~$y$ are distinct and~$u$ is their unique lowest common ancestor in~$N$.
\end{observation}

\begin{observation}\label{obs:retictrinet}
Suppose that a network~$N$ contains a tree path from a reticulation~$r$ to a leaf~$x$. Then,~$r$ is the only reticulation with a tree path to~$x$. Moreover, suppose that~$p_1$ and~$p_2$ are the parents of~$r$ and that there is a directed path from~$p_1$ to~$p_2$ and a tree-path from~$p_1$ to a leaf~$y$. In addition, suppose that~$P$ is a trinet exhibited by~$N$ that contains~$x$ and~$y$. Then,~$P$ contains~$r$ and a tree path from~$r$ to~$x$.
\end{observation}

Notice that, in Observation~\ref{obs:retictrinet}, the presence of leaf~$y$ in trinet~$P$ ensures that, in the process of obtaining~$P$ from~$N$, the incoming arcs of~$r$ do not become parallel arcs, which would have to be suppressed. We are now ready to prove the main result of this section.

\begin{theorem}\label{thm:treechild}
Every binary tree-child network~$N$ on~$X$, with~$|X|\geq 3$, is encoded by its set of trinets~$Tn(N)$.
\end{theorem}
\begin{proof}
The proof is by induction on the level~$k$ of the network. The induction basis for~$k=1$ has been shown in~\cite{huber2011encoding}. 

Let~$k\geq 2$ and assume that every binary, tree-child, level-$(k-1)$ network with at least three leaves is encoded by its trinets. Let~$N$ be a binary tree-child level-$k$ network on~$X$, with~$|X|\geq 3$, and let $\cT=Tn(N)$. If~$|X|=3$, the theorem is obviously true, hence we can assume $|X|\geq 4$. By Theorem~\ref{thm:bcc}, we may assume that~$N$ is a simple level-$k$ network, i.e. it has a single nontrivial biconnected component~$B$ and no cut-arcs except for the ones leaving~$B$. Consequently, for each cut-arc~$(u,v)$ of $N$, the vertex~$v$ is a leaf. 

Now, let~$N'$ be any recoverable network on~$X$ exhibiting~$\cT$. We will show that~$N'=N$. Suppose that $x$ is a leaf of~$N$ at maximum distance from the root. 

We first claim that the parent of~$x$ is a reticulation~$r$ such that there is no arc between the parents of~$r$. To see this, first assume that~$r$ is not a reticulation. Then it has some other child~$s$, which must be a leaf because otherwise there would be a tree-path from~$s$ to some leaf at greater distance from the root than~$x$. However, in that case, the arc entering~$r$ would be a cut-arc, which is not possible because~$N$ is a simple level-$k$ network. Hence,~$r$ is a reticulation. Now let~$p_1$ and~$p_2$ be the parents of~$r$ and assume that there is an arc $(p_1,p_2)$. Since~$p_2$ is not a reticulation by the tree-child property, it has a second child~$s$. Observe that, by the tree-child property,~$s$ cannot be a reticulation. Moreover, if~$s$ is not a leaf then it has two children, which must be leaves because~$x$ has maximum distance from the root. But, this is not possible because then the arc entering~$s$ would be a cut-arc. Hence,~$s$ is a leaf. However, then there are only two leaves below~$p_1$. Since ${|X|\geq 4}$, there must be some arc entering~$p_1$ and this is a cut-arc. This is again a contradiction to the fact that~$N$ is a simple level-$k$ network. Thus, we conclude that the parent~$r$ of~$x$ is a reticulation and that there is no arc between the parents of~$r$, as claimed.

Now, let~$\cT^*$ be the result of removing all trinets containing~$x$ from~$\cT$ and let~$N^*$ be the result of removing~$x$ from~$N$ and ``cleaning up'' the network by repeatedly deleting unlabelled outdegree-0 vertices and indegree-0 outdegree-1 vertices and suppressing indegree-1 outdegree-1 vertices, until a valid network is obtained. Note that it is not necessary to suppress parallel arcs and redundant biconnected components because these cannot arise by the described modifications since~$N$ is a tree-child network. Moreover, it can easily be seen that~$N^*$ is again a tree-child network and has level $k-1$. Since~$N$ has at least four leaves,~$N^*$ has at least three leaves. Hence, by induction, $N^*$ is encoded by its trinets. It follows that removing~$x$ from~$N'$, and cleaning up in the same way as we did in~$N$, also gives~$N^*$. Hence it only remains to show that the location of~$x$ in~$N$ and~$N'$ is the same.

\begin{figure}[t]
    \centering
    \vspace{-.3cm}
    \includegraphics[scale=.65]{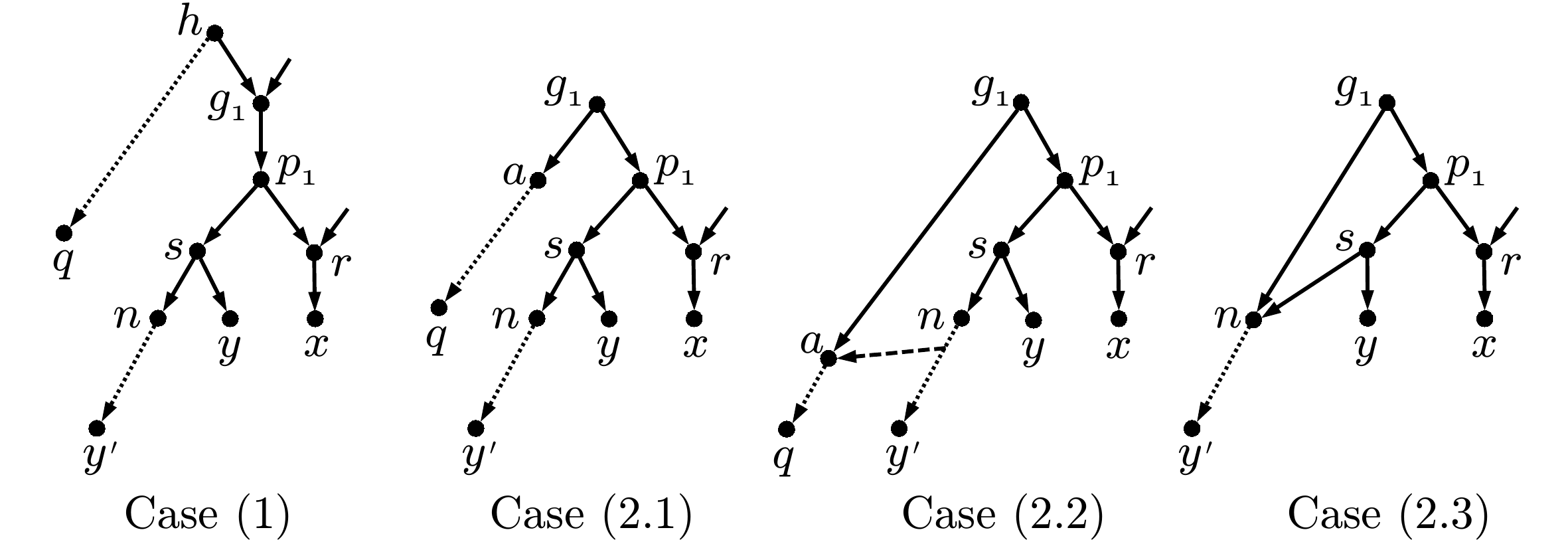}
    \vspace{-.3cm}
    \caption{Illustration of the main cases in the proof of Theorem~\ref{thm:treechild}. Dotted arcs denote tree paths (directed paths not passing through reticulations). The dashed arc from below~$n$ to~$a$ in Case~(2.2) denotes a directed path that might contain reticulations. Case~(3) is very similar to Case~(1).\label{fig:treechildthm}}
\end{figure}

To this end, consider again the reticulation~$r$ in~$N$ of which~$x$ is the child and the parents~$p_1$ and~$p_2$ of~$r$ in~$N$ and in~$N'$. We need to show that the location of~$p_1$ and~$p_2$ is the same in~$N$ and~$N'$. We consider~$p_1$ and note that the same arguments can be applied to~$p_2$. By the tree-child property, there is a tree path in~$N$ from~$p_1$ to a leaf~$y$. Thus,~$p_1$ has outdegree~2 and hence it is not a reticulation. We distinguish three cases: (1) the parent~$g_1$ of~$p_1$ in~$N$ is a reticulation, (2) the parent~$g_1$ of~$p_1$  in~$N$ has outdegree~2, and (3) $p_1$ is the root~$\rho$ of~$N$. See Figure~\ref{fig:treechildthm} for some illustrations of these cases.

Case (1): The parent~$g_1$ of~$p_1$ in~$N$ is a reticulation. Let~$h$ be a parent of~$g_1$ such that there exists no directed path from~$h$ to the other parent of~$g_1$. Then there is a tree path from~$h$ to some leaf~$q$. Observe that~$q$ and~$y$ must be distinct and that~$h$ must be their unique lowest common ancestor in~$N$ by Observation~\ref{obs:lcaunique}. Moreover, the same holds in~$N^*$ and consequently in~$N'$ because removing leaf~$x$ and cleaning up as specified does not affect lowest common ancestors of other leaves. Consider the trinet~$P_{xyq}\in\cT$ on $\{x,y,q\}$. In~$P_{xyq}$, $h$ is also the unique lowest common ancestor of~$q$ and~$y$ by Observation~\ref{obs:lcatrinet}. Moreover, in~$P_{xyq}$ also, $x$ is below the reticulation-child (which we can also call~$g_1$) of~$h$. This means that, because~$N'$ exhibits~$P_{xyq}$, $p_1$ is below~$g_1$ in~$N'$. We will show that~$p_1$ is in fact the child of~$g_1$ in $N'$ (just as it is in $N$).

Let~$s$ be the child of~$p_1$ other than~$r$, in~$N$. Note that~$s$ cannot be a reticulation by the tree-child property. If~$s$ is a leaf, then $s=y$ is the child of~$g_1$ in~$N^*$ and, since we know that~$p_1$ is below~$g_1$ in~$N'$,~$p_1$ can only be the child of~$g_1$ in~$N'$. Now suppose that~$s$ is not a leaf in~$N$. Then it has two children, and one of these children is~$y$ because otherwise~$y$ would be at greater distance from the root than~$x$. Let~$n$ be the other child of~$x$. (Note that~$n$ may or may not be a reticulation but that~$n$ cannot be equal to~$r$ because there is no arc between the parents of~$r$ (by the choice of~$x$).) There is a tree path from~$n$ to some other leaf~$y'$. By Observation~\ref{obs:lcaunique},~$s$ is the unique lowest common ancestor of~$y$ and~$y'$ in~$N$, in~$N'$ and in~$N^*$. In view of the trinet on $\{y,y',x\}$, it follows that $LCA(x,y)\leq_{N'} LCA(y,y')$, i.e. $p_1\leq_{N'} s$. Since~$s$ is below~$p_1$ and we already know that~$p_1$ is below~$g_1$, we conclude that~$p_1$ is the child of~$g_1$ in~$N'$, as required.

Case (2): The parent~$g_1$ of~$p_1$ has outdegree~2. Then~$g_1$ has some child~$a$ other than~$p_1$. Note that~$a$ cannot be equal to~$r$ because there is no arc between the parents of~$r$. From~$a$ there is a tree path to some leaf~$q$. As before, let~$s$ be the child of~$p_1$ other than~$r$, in~$N$. There is again a tree path from~$s$ to some leaf~$y$. If~$s$ is a leaf, then again~$s=y$ and, in view of the trinet on~$\{x,y,q\}$, $p_1$ is the parent of~$y$ in~$N'$. Now we distinguish three subcases: (2.1) $a\neq n$ and there is no directed path from~$s$ to~$a$, (2.2) $a\neq n$ and there \emph{is} a directed path from~$s$ to~$a$ and (2.3) $a=n$ (and consequently there is no directed path from~$s$ to~$a$).

In Case (2.1),~$q$ and~$y$ are distinct and~$g_1$ is their unique lowest common ancestor by Observation~\ref{obs:lcaunique}. We can now use similar arguments as in Case (1) to show that, in~$N'$,~$p_1$ is the child of~$g_1$ on the directed path to~$y$.

In Case (2.2),~$a$ is the unique reticulation from which there is a tree path to~$q$ and~$g_1$ is its parent from which there is a directed path to the other parent, in~$N,N^*,N'$ and in the trinet on~$\{q,y,x\}$, by Observation~\ref{obs:retictrinet}. In view of this trinet,~$p_1$ is below~$g_1$, on the directed path to~$y$. We can now use similar arguments as in Case (1) to show that, in~$N'$,~$p_1$ is the child of~$g_1$ on the directed path to~$y$.

In Case (2.3),~$s$ is the unique lowest common ancestor of~$y$ and~$y'$ in~$N,N^*,N'$ and in any trinet containing~$y,y'$ by Observations~\ref{obs:lcatrinet} and \ref{obs:lcaunique}. Also,~$g_1$ is the parent of~$s$ in~$N^*$. In view of the trinet on~$\{y,y',x\}$, $p_1$ has to be the parent of~$s=LCA(y,y')$ in~$N'$. Thus, the location of~$p_1$ is the same in~$N$ and~$N'$.

Case (3): $p_1$ is the root~$\rho$ of~$N$. We define~$s,n,y,y'$ as before. In this case,~$s$ is the root of~$N^*$. Hence,~$s$ cannot be a leaf. We can argue as in Case (1), concluding that~$p_1$ is the root of~$N'$.

After applying exactly the same arguments to~$p_2$ as we did to~$p_1$, it follows that, in all cases, the location of~$p_1$ and~$p_2$ is the same in~$N$ and~$N'$. Hence, the location of~$x$ is the same in~$N$ and~$N'$. It follows that~$N=N'$.
\end{proof}

\section{Discussion}\label{sec:postlude}

We have proven that binary, recoverable level-2 and binary tree-child networks
are encoded by their trinets, using 
two distinct methods of proof.
We expect that our results could also 
hold for non-binary networks, and it would be of interest 
to verify this.

For settling the question if all recoverable phylogenetic networks are encoded by their trinets, the decomposition theorems in Section~\ref{sec:decomp} will be useful since they essentially show that it is sufficient to answer this question for simple networks (i.e. networks having no cut-arcs apart from pendant arcs).

The proof for level-2 networks might be extended 
to show that higher level networks are encoded
by their trinets (or be used to provide a counter-example). 
However, a new technique would have to be developed 
for~$k\geq 4$ since, for such~$k$, there exist 
level-$k$ networks that have no crucial trinets.
Another difficulty is that 
the number of generators for level-$k$ networks
grows very rapidly (the number of level-$k$
generators is at least $2^{k-1}$ \cite{Gambette2009structure})
making a similar case analysis impossible in general.
To prove that tree-child networks 
are encoded by trinets, we heavily depended on special 
properties of such networks, and we have not
been able to find a way to extend our proof to 
even slightly more general networks (e.g.   
reticulation-visible networks~\cite{HusonRuppScornavacca10}).

We note that a natural extension to the definition 
of ``exhibit'' is to define it as in 
Observation~\ref{obs:exhibit} but to 
suppress not only \emph{strongly} redundant biconnected components, 
but \emph{all} redundant biconnected components. 
If one then changes the definition of ``recoverable'' 
accordingly (i.e. to not having any redundant 
biconnected components), then it can be checked 
that all proofs in this paper still hold. 
This could be relevant when reconstructing 
phylogenetic networks via trinets, because 
the number of recoverable trinets then 
becomes bounded by a function of~$k$.

It is also worth noting that 
Theorem~\ref{thm:bcc}, Theorem~\ref{lem:simplelevel2} and 
Theorem~\ref{thm:treechild} can be combined 
to provide the following more general result.

\begin{corollary}\label{cor:comb}
If~$X$ is a finite set with~$|X|\geq 3$ and~$\cN$ 
is the set of binary recoverable phylogenetic 
networks~$N$ on~$X$ for which each biconnected component of~$N$ either
\begin{itemize}
\item has at most two reticulations; or
\item is tree-child; or
\item has at most three outgoing cut-arcs,
\end{itemize}
then each~$N\in\cN$ is encoded by~$Tn(N)$.
\end{corollary}

In addition, we note that our results also yield some 
new metrics on level-2 and tree-child networks.
These are of potential intesest since several metrics have been 
recently developed for special classes of networks
(see e.g. \cite{CRV08a,CLRV09a,CLRV09b,CLR11,huber2011encoding}).
More specifically, Corollary~\ref{cor:comb} immediately 
implies the following result (where~$\Delta$ 
denotes the symmetric difference of two sets).

\begin{corollary} \label{metric}
If~$X$ is a finite set with~$|X|\geq 3$ and~$\cN$ is
as in Corollary~\ref{cor:comb}, then the 
map $d: {\mathcal N} \times {\mathcal N} \to \R$ defined by
$$
d(N,N') := |Tn(N) \Delta Tn(N')|,
$$
for all $N,N'\in {\mathcal N}$, is a metric on 
${\mathcal N}$.
\end{corollary}

Finally, it could be of some interest to 
study some algorithmic issues related to 
the results that we have presented.
For example, it would be interesting 
to know whether or not it is possible to reconstruct a 
recoverable (level-2 or tree-child) network from a 
set of trinets in polynomial time. Hopefully
shedding light on this and related 
complexity problems could help provide new algorithms for 
constructing phylogenetic networks. 

\bibliographystyle{amsplain}
\bibliography{trinetsbibliography}

\providecommand{\bysame}{\leavevmode\hbox to3em{\hrulefill}\thinspace}
\providecommand{\MR}{\relax\ifhmode\unskip\space\fi MR }
\providecommand{\MRhref}[2]{%
  \href{http://www.ams.org/mathscinet-getitem?mr=#1}{#2}
}
\providecommand{\href}[2]{#2}
\begin{thebibliography}{10}

\bibitem{BaroniEtAl2004}
M.~Baroni, C.~Semple, and M.~Steel, \emph{A framework for representing
  reticulate evolution}, Annals of Combinatorics \textbf{8} (2004), 391--408.

\bibitem{BGJ10}
J.~Byrka, S.~Guillemot, and J.~Jansson, \emph{New results on optimizing rooted
  triplets consistency}, Discrete Applied Mathematics \textbf{158} (2010),
  1136--1147.

\bibitem{CRV08a}
G.~Cardona, M.~Llabr\'es, F.~Rossell\'o, and G.~Valiente, \emph{A distance
  metric for a class of tree-sibling phylogenetic networks}, Bioinformatics
  \textbf{24} (2008), 1481--1488.

\bibitem{CLRV09a}
\bysame, \emph{Metrics for phylogenetic networks i: Generalization of the
  robinson-foulds metric}, IEEE/ACM Transactions in Computational Biology and
  Bioinformatics \textbf{6} (2009), 46--61.

\bibitem{CLRV09b}
\bysame, \emph{Metrics for phylogenetic networks ii: Nodal and triplets
  metrics}, IEEE/ACM Transactions in Computational Biology and Bioinformatics
  \textbf{6} (2009), 454--469.

\bibitem{CLRV2010b}
\bysame, \emph{Path lengths in tree-child time consistent hybridization
  networks}, Information Sciences \textbf{180} (2010), no.~3, 366--383.

\bibitem{CLR11}
\bysame, \emph{Comparison of galled trees}, IEEE/ACM Transactions on
  Computational Biology and Bioinformatics \textbf{8} (2011), 410--427.

\bibitem{Cardona2007}
G.~Cardona, F.~Rossell\'o, and G.~Valiente, \emph{Comparison of tree-child
  phylogenetic networks}, IEEE/ACM Transactions on Computational Biology and
  Bioinformatics \textbf{6} (2009), no.~4, 552--569.

\bibitem{DressBasic}
A.~Dress, K.T. Huber, J.~Koolen, V.~Moulton, and A.~Spillner, \emph{Basic
  phylogenetic combinatorics}, Cambridge University Press, 2012.

\bibitem{FH10}
J.~Fischer and D.~Huson, \emph{New common ancestor problems in trees and
  directed acyclic graphs}, Information processing letters \textbf{110} (2010),
  331--335.

\bibitem{Gambette2009structure}
P.~Gambette, V.~Berry, and C.~Paul, \emph{The structure of level-k phylogenetic
  networks}, Proceedings of the 20th Annual Symposium on Combinatorial Pattern
  Matching (Berlin, Heidelberg), CPM '09, Springer-Verlag, 2009, pp.~289--300.

\bibitem{GBP2012}
\bysame, \emph{Quartets and unrooted phylogenetic networks}, Journal of
  Bioinformatics and Computational Biology \textbf{10} (2012), no.~4, 1250004.

\bibitem{GambetteHuber2012}
P.~Gambette and K.T. Huber, \emph{On encodings of phylogenetic networks of
  bounded level}, Journal of Molecular Biology \textbf{65} (2012), no.~1,
  157--180.

\bibitem{HabibTo2012}
M.~Habib and T.-H. To, \emph{Constructing a minimum phylogenetic network from a
  dense triplet set}, Journal of Bioinformatics and Computational Biology
  \textbf{10} (2012), no.~5.

\bibitem{lev1athan}
K.T. Huber, L.J.J.~van Iersel, S.M. Kelk, and R.~Suchecki, \emph{A practical
  algorithm for reconstructing level-1 phylogenetic networks}, IEEE/ACM
  Transactions on Computational Biology and Bioinformatics \textbf{8} (2011),
  no.~3, 635--649.

\bibitem{huber2011encoding}
K.T. Huber and V.~Moulton, \emph{Encoding and constructing 1-nested
  phylogenetic networks with trinets}, Arxiv preprint arXiv:1110.0728 (2011).

\bibitem{HusonRuppScornavacca10}
D.H. Huson, R.~Rupp, and C.~Scornavacca, \emph{Phylogenetic networks: Concepts,
  algorithms and applications}, Cambridge University Press, 2011.

\bibitem{RECOMB2008}
L.J.J.~van Iersel, J.C.M. Keijsper, S.M. Kelk, L.~Stougie, F.~Hagen, and
  T.~Boekhout, \emph{Constructing level-2 phylogenetic networks from triplets},
  Research in Computational Molecular Biology (RECOMB), Lecture Notes in
  Bioinformatics, vol. 4955, 2008, pp.~464--476.

\bibitem{simplicityAlgorithmica}
L.J.J.~van Iersel and S.M. Kelk, \emph{Constructing the simplest possible
  phylogenetic network from triplets}, Algorithmica \textbf{60} (2009), 1--29.

\bibitem{twotrees}
\bysame, \emph{When two trees go to war}, Journal of Theoretical Biology
  \textbf{269} (2011), no.~1, 245--255.

\bibitem{ISS2010b}
L.J.J.~van Iersel, C.~Semple, and M.~Steel, \emph{Locating a tree in a
  phylogenetic network}, Information Processing Letters \textbf{110} (2010),
  no.~23, 1037--1043.

\bibitem{JanssonEtAl2006}
J.~Jansson, N.B. Nguyen, and W-K. Sung, \emph{Algorithms for combining rooted
  triplets into a galled phylogenetic network}, SIAM Journal on Computing
  \textbf{35} (2006), no.~5, 1098--1121.

\bibitem{JNST06}
G.~Jin, L.~Nakhleh, S.~Snir, and T.~Tuller, \emph{Maximum likelihood of
  phylogenetic networks}, Bioinformatics \textbf{22} (2006), 2604--2611.

\bibitem{JNST09}
\bysame, \emph{Parsimony score of phylogenetic networks: Hardness results and a
  linear-time heuristic}, IEEE/ACM Transactions on Computational Biology and
  Bioinformatics \textbf{6} (2009), 495--505.

\bibitem{davidbook}
D.A. Morrison, \emph{Introduction to phylogenetic networks}, RJR Productions,
  Uppsala, 2011.

\bibitem{N11}
L.~Nakhleh, \emph{Evolutionary phylogenetic networks: Models and issues},
  Problem Solving Handbook in Computational Biology and Bioinformatics (L~S
  Heath and N~Ramakrishnan, eds.), Springer Berlin / Heidelberg, 2011.

\bibitem{SempleSteel2003}
C.~Semple and M.~Steel, \emph{Phylogenetics}, Oxford University Press, 2003.
  \MR{MR2060009 (2005g:92024)}

\bibitem{Willson2010}
S.J. Willson, \emph{Regular networks can be uniquely constructed from their
  trees}, IEEE/ACM Transactions on Computational Biology and Bioinformatics
  \textbf{8} (2010), no.~3, 785--796.

\bibitem{Willson2012b}
\bysame, \emph{Tree-average distances on certain phylogenetic networks have
  their weights uniquely determined}, Algorithms for Molecular Biology
  \textbf{7} (2012), no.~13.

\end{thebibliography}

\end{document}